\documentclass[review]{elsarticle}

\usepackage{lineno,hyperref}
\modulolinenumbers[5]

\journal{Information Processing Letters}

\usepackage{graphicx}
\usepackage{enumerate}
\usepackage{setspace}
\usepackage{multirow}
\usepackage{graphics}
\usepackage{amsmath}
\usepackage{amssymb}
\usepackage{array}
\usepackage{color}
\usepackage{url}
\usepackage{float}

\usepackage[ruled,vlined]{algorithm2e}

\begin{document}

\renewcommand{\thepart}{\Roman{part}}
\hyphenation{long-est sub-string insert-ancestor non-overlapping ex-tend-ible}

\newtheorem{theorem}{Theorem}
\newtheorem{lemma}{Lemma}
\newtheorem{proposition}{Proposition}
\newtheorem{defn}{Definition}
\newtheorem{example}{Example}
\newtheorem{corollary}{Corollary}

\newenvironment{dig}{\\ [6pt]\noindent {\bf Digression}}{~$\Box$\\ [6pt]\indent}
\newenvironment{dig1}{\noindent {\bf Digression}}{~$\Box$\\ [6pt]\indent}
\newenvironment{proof}{\noindent {\bf Proof}\ }{~$\Box$\\ [6pt]\indent}
\newtheorem{alg}{\hspace{1.3in} Algorithm}
\def\itbf#1{\textit{\textbf{#1}}}

\newif\ifj
\newif\ifComments
\Commentstrue

\def\s#1{\mbox{\boldmath $#1$}}
\def\S{\mathcal S}
\def\pref#1{\mbox{pref(\s{#1})}}
\def\suff#1{\mbox{suff(\s{#1})}}
\def\head{\mbox{\rm head}}
\def\tail{\mbox{\rm tail}}
\def\CFL{\mbox{\rm CFL}}
\def\NF{\mbox{\rm NF}}
\def\ST{\mbox{\rm ST}}
\def\SA{\mbox{\rm SA}}
\def\IN{\mbox{\rm IN}}
\def\LZ{\mbox{\rm LZ}}
\def\INC{\mbox{\rm INC}}
\def\LCS{\mbox{\rm LCS}}
\def\MSP{\mbox{\rm MSP}}
\def\LCA{\mbox{\rm LCA}}
\def\LCF{\mbox{\rm LCF}}
\def\LCP{\mbox{\rm LCP}}
\def\lcp{\mbox{\rm lcp}}
\def\POS{\mbox{\rm POS}}
\def\LEN{\mbox{\rm LEN}}
\def\UMFF{\mbox{\rm UMFF}}
\def\BELONGS{\mbox{\rm BELONGS}}
\def\pop{\mbox{\tt pop}}
\def\TRUE{\mbox{\tt TRUE}}
\def\FALSE{\mbox{\tt FALSE}}
\def\push{\mbox{\tt push}}
\def\lcs{\mbox{\rm lcs}}
\def\exp{\mbox{\rm exp}}
\def\prb{\mbox{\rm prob}}
\def\NDFA{\mbox{\rm NDFA}}
\def\DFA{\mbox{\rm DFA}}
\def\ACFA{\mbox{\rm ACFA}}
\def\CWFA{\mbox{\rm CWFA}}
\def\BYNFA{\mbox{\rm BYNFA}}
\def\NRE{\mbox{\rm NRE}}
\def\NE{\mbox{\rm NE}}
\def\E{\mbox{\rm E}}
\def\LOC{\mbox{\rm LOC}}
\def\BRANCH{\mbox{\rm BRANCH}}
\def\ov#1{\mbox{$\overline{\s{t_{#1}}}$}}
\def\OV#1{\mbox{$\overline{\lambda_{#1}}$}}
\def\t#1{\mbox{\boldmath $\scriptstyle#1$}}
\def\stirl#1#2{\mbox{$\genfrac{\{}{\}}{0pt}{}{#1}{#2}$}}
\def\stir#1#2{\genfrac{\{}{\}}{0pt}{}{#1}{#2}}
\def\binom#1#2{\genfrac{(}{)}{0pt}{}{#1}{#2}}
\def\comma{,\mbox{\hspace{4pt}}}
\def\semi{;\mbox{\hspace{4pt}}}
\def\:{\mbox{\ :\ }}
\def\floor#1{\lfloor #1 \rfloor}
\def\bfloor#1{\big\lfloor #1 \big\rfloor}
\def\Bfloor#1{\Big\lfloor #1 \Big\rfloor}
\def\ceil#1{\lceil #1 \rceil}
\def\bceil#1{\big\lceil #1 \big\rceil}
\def\+{\!+\!}
\def\-{\!-\!}
\def\plmi{\!\pm\!}
\def\m{\!-\!}
\def\u#1{\underline{#1}}
\def\o#1{\overline{#1}}

\def\bproc{{\bf procedure\ }}
\def\bfunc{{\bf function\ }}
\def\bvar{{\bf var\ }}
\def\barray{{\bf array\ }}
\def\bof{{\bf of\ }}
\def\bfor{{\bf for\ }}
\def\bto{{\bf to\ }}
\def\bdownto{{\bf downto\ }}
\def\bwhile{{\bf while\ }}
\def\brep{{\bf repeat\ }}
\def\buntil{{\bf until\ }}
\def\band{{\bf and\ }}
\def\bor{{\bf or\ }}
\def\bdo{{\bf do\ }}
\def\bif{{\bf if\ }}
\def\bthen{{\bf then\ }}
\def\belse{{\bf else\ }}
\def\belsif{{\bf elsif\ }}
\def\bnot{{\bf not\ }}
\def\bgoto{{\bf goto\ }}
\def\breturn{{\bf return\ }}
\def\boutput{{\bf output\ }}
\def\la{\leftarrow}
\def\ra{\rightarrow}
\def\llra{\longleftrightarrow}
\def\q{\quad}
\def\qq{\qquad}
\def\com#1{\hspace{6pt}{\bf ---}\hspace{6pt}{\sl #1}}
\def\rem#1{\hspace{24pt}{\sl #1}}
\def\pref(#1,#2){$#1$ is a prefix of $#2$}
\def\suff(#1,#2){$#1$ is a suffix of $#2$}
\def\reg(#1,#2){$#2$ is $#1$-regular}
\def\notreg(#1,#2){$#2$ is not $#1$-regular}

\begin{frontmatter}


\title{Enhanced string factoring from alphabet orderings\tnoteref{mytitlenote,myfootnote}}
\tnotetext[mytitlenote]{The authors were part-funded by the European Regional Development Fund through the Welsh Government}
\tnotetext[myfootnote]{A preliminary version of this paper was accepted as a poster in IWOCA 2018 (International Workshop on Combinatorial Algorithms)}

\author[mymainaddress]{Amanda Clare\corref{mycorrespondingauthor}} 
\author[mymainaddress,mysecondaryaddress,mythirdaddress]{Jacqueline W. Daykin}
\address[mymainaddress]{Department of Computer Science, Aberystwyth University, SY23 3DB, UK\\
\{afc,jwd6\}@aber.ac.uk}
\address[mysecondaryaddress]{Department of Informatics,
King's College London, WC2B 4BG, UK\\
jackie.daykin{@}kcl.ac.uk}
\address[mythirdaddress]{Department of Information Science, Stellenbosch
University, South Africa} 
\cortext[mycorrespondingauthor]{Corresponding author}

\begin{abstract}
In this note we consider the concept of alphabet ordering in the context of string factoring. We propose a greedy-type algorithm which produces Lyndon factorizations with small numbers of factors along with a modification for large numbers of factors. For the technique we introduce the Exponent Parikh vector. Applications and research directions derived from circ-UMFFs are discussed.
\end{abstract}

\begin{keyword}
alphabet order  \sep big data \sep circ-UMFF \sep factor \sep factorization  \sep greedy algorithm \sep lexicographic orderings \sep Lyndon word \sep sequence alignment \sep string


\end{keyword}

\end{frontmatter}


\section{Introduction}

Factoring strings is a powerful form of the divide and conquer problem-solving paradigm for strings or words. Notably the Lyndon factorization is both efficient to compute and useful in practice \cite{CFL-58, Du-83}.
We study the effect of an alphabet's order on the number of factors in a Lyndon factorization and propose a greedy-type algorithm for assigning an order to the alphabet. In addition, we formalize the distinction between the sets of Lyndon and co-Lyndon words as avenues for alternative string factorizations. More generally, circ-UMFFs  provide the opportunity for achieving further diversity with string factors \cite{DD-08, DDS-09}.

\subsection{Notation}
\label{subsec-notation}

Given an integer $n \ge 1$ and a nonempty set of symbols $\Sigma$ (bounded or unbounded),
a \itbf{string of length n}, equivalently \itbf{word}, over $\Sigma$ takes the form $\s{x} = x_1 ... x_n$ with each $x_i \in \Sigma$.
For brevity, we write $\s{x} = \s{x}[1..n]$ and we let $\s{x}[i]$ denote the $i$-th symbol of \s{x}.
The length $n$ of a string \s{x} is denoted by $|\s{x}|$. The set
$\Sigma$ is called an \itbf{alphabet} whose members are \itbf{letters} or \itbf{characters}, and
$\Sigma^+$ denotes the set of all nonempty finite strings over $\Sigma$.
The \itbf{empty string} of length zero is
denoted \s{\varepsilon}; we write $\Sigma^* = \Sigma^+ \cup \{\s{\varepsilon}\}$.
A string \s{w} is called a \itbf{factor} of $\s{x}[1..n]$ if and only if $\s{w} = \s{x}[i..j]$ for $1 \le i \le j \le n$. If $\s{x} = \s{uv}$,
then \s{vu} is said to be a \itbf{rotation} (\itbf{cyclic shift} or \itbf{conjugate})
of \s{x}. 
A string \s{x} is said to be a \itbf{repetition} if and only if it has a factorization
$\s{x} = \s{u}^k$ for some integer $k > 1$;
otherwise, \s{x} is said to be \itbf{primitive}.
For a string \s{x}, the reversed string \s{\overline{x}} is defined as
$\s{\overline{x}} = \s{x}[n]\s{x}[n\- 1]\cdots\s{x}[1]$.
A string which is both a proper prefix and a proper suffix of a nonempty string \s{x} is called a \itbf{border} of \s{x}.

If $\Sigma$ is a totally ordered alphabet then \itbf{lexicographic ordering} (\itbf{lexorder}) $\s{u} < \s{v}$ with $\s{u},\s{v} \in \Sigma^+$
is defined if and only if either \s{u} is a proper prefix of
\s{v}, or $\s{u}=\s{r}a\s{s}$, $\s{v}=\s{r}b\s{t}$ for some $a,b \in \Sigma$ such that $a<b$ and for some $\s{r},\s{s},\s{t} \in \Sigma^*$.
We call the ordering $\prec$ based on lexorder of reversed strings \itbf{co-lexicographic ordering} (\itbf{co-lexorder}).

\section{Unique Maximal Factorization Families (UMFFs)}
\label{sec-UMFFs}

A subset $\mathcal{W} \subseteq \Sigma^+$ is a
\itbf{factorization family} (FF) if and only if for {\rm every} nonempty string
\s{x} on $\Sigma$ there exists a factorization of \s{x} over $\mathcal{W}$, $F_{\mathcal{W}}(\s{x})$.
If every factor of 
$F_{\mathcal{W}}(\s{x})$ is maximal (\itbf{max}) with respect to $\mathcal{W}$ then the factorization is said to be max, and hence must be unique. So if $\mathcal{W}$ is an FF on an alphabet $\Sigma$ then $\mathcal{W}$ is a \itbf{unique maximal factorization family} (UMFF)
if 
there exists a max factorization $F_{\mathcal{W}}(\s{x})$
for {\rm every} string $\s{x} \in \Sigma^+$ -- for this theory see \cite{DD-08,DDS-09}.

An UMFF $\mathcal{W}$ 
is a \itbf{circ-UMFF} if 
it contains exactly one rotation of every primitive string $\s{x} \in \Sigma^+$.
The classic and foundational 
circ-UMFF is the set of Lyndon words, which we denote $\mathcal{L}$, where the rotation chosen 
is the one that is least in the lexorder derived from an ordering of the letters of the alphabet $\Sigma$
(\cite{CFL-58,Du-83,
DD-08}).
Subsequently, the co-Lyndon circ-UMFF was formed in \cite{DDS-09} consisting of those words which are least amongst their rotations in co-lexorder.

Every circ-UMFF $\mathcal{W}$ yields a strict order relation, the \itbf{$\mathcal{W}$-order}: if $\mathcal{W}$ contains strings \s{u}, \s{v} and \s{uv} then $\s{u} <_{\mathcal{W}} \s{v}$. 
For the Lyndon circ-UMFF, its specific $\mathcal{W}$-order is lexorder:

\begin{theorem}
\label{thm-Lynorder}
(Duval \cite{Du-83}) Let $\mathcal{L}$ be the set of Lyndon words, and suppose $\s{u}, \s{v} \in \mathcal{L}$. Then $\s{uv} \in \mathcal{L}$
if and only if $\s{u}$ comes before $\s{v}$ in lexorder.
\end{theorem}

It was observed in \cite{DDS-09} that the analogue of Theorem~\ref{thm-Lynorder} does not hold for every circ-UMFF. 
We show here that 
the respective orders for the sets of co-Lyndon words and words in co-lexorder are always distinct.

\begin{lemma}
\label{lem-coLynorder}
Let $\mbox{co-}\mathcal{L}$ be the set of co-Lyndon words, and suppose $\s{u}, \s{v} \in \mbox{co-}\mathcal{L}$. Then $\s{uv} \in \mbox{co-}\mathcal{L}$ if and only if $\s{v}$ comes before $\s{u}$ in co-lexorder. 
\end{lemma} 

\begin{proof}
Since $\s{u}, \s{v} \in \mbox{co-}\mathcal{L}$ then $\s{\overline{u}}, \s{\overline{v}} \in \mathcal{L}$. If $\s{v} \prec \s{u}$ in co-lexorder then $\s{\overline{v}} < \s{\overline{u}}$ in lexorder. Applying Theorem \ref{thm-Lynorder} we have $\s{\overline{v}} \s{\overline{u}} \in \mathcal{L}$ and hence $\s{uv} \in \mbox{co-}\mathcal{L}$. 
Next if $\s{uv} \in \mbox{co-}\mathcal{L}$ then it must be primitive and border-free \cite{DD-08}. Thus $\s{u} \neq \s{v}$ which gives rise to two cases. Suppose first that $\s{u} \prec \s{v}$. If \s{u} is a proper suffix of \s{v} then $\s{uv} =  \s{uwu}$ for some $\s{w} \neq \varepsilon$ contradicting the border-free property. Otherwise, with $|\s{u}| = n$ there is some largest $j$, $1 \le j \le n$, such that $\s{u}[j] \neq \s{v}[j]$. If $\s{u}[j] < \s{v}[j]$ then $\s{vu} \prec \s{uv}$ contradicting $\s{uv} \in \mbox{co-}\mathcal{L}$. We conclude that $\s{u}[j] > \s{v}[j]$, and so $\s{v} \prec \s{u}$ as required. 
\end{proof}

The sets of Lyndon and co-Lyndon words are distinct and almost disjoint.

\begin{lemma}
\label{lem-noteq}
$\mathcal{L} \neq \mbox{co-}\mathcal{L}$ and $\mathcal{L} \cap \mbox{co-}\mathcal{L} = \Sigma$.
\end{lemma}

\begin{proof}
Let $\s{v} \in \mathcal{L}$ and $\s{w} \in \mbox{co-}\mathcal{L}$ with $|\s{v}|, |\s{w}| > 2$. Then $\s{v}$ starts with some letter $\alpha$ which is minimal in $\s{v}$. Since $\s{v}$ is border-free then it ends with some $\beta$ where $\alpha < \beta$. Similary, $\s{w}$ starts $\gamma$ and ends $\delta$, where $\gamma > \delta$. Therefore $\s{v} \neq \s{w}$. Finally, every circ-UMFF contains the alphabet $\Sigma$ as expressed in \cite{DD-08,DDS-09}.
\end{proof}
The following result generalizes the Lyndon factorization theorem \cite{CFL-58} and is a key to further applications of string decomposition.

\begin{theorem}
\label{thm-faxmono}
\cite{DD-08}
Let ${\mathcal{W}}$ be a circ-UMFF and suppose $\s{x} = \s{u_1u_2\cdots u_m}$,
with each $\s{u_j} \in {\mathcal{W}}$.
Then $F_{\mathcal{W}}(\s{x}) = \s{u_1u_2}\cdots\s{u_m}$ if and only if
$\s{u_1} \ge_{\mathcal{W}} \s{u_2} \ge_{\mathcal{W}}  ... \ge_{\mathcal{W}} \s{u_m}$.
\end{theorem}

\section{Alphabet ordering}
\label{sec-order}

Suppose the goal is to optimize 
a Lyndon factorization according to minimizing or maximizing the number of factors. For this we consider choosing the order of the letters in the - assumed unordered - alphabet so as to influence the number of factors. To illustrate, consider the string $\s{x} = abcabcdabcaba$. 
If $\Sigma = \{a<b<c<d\}$, then $F_{\mathcal{L}}(\s{x}) = abcabcd \ge abc \ge ab \ge a$. Whereas, if we choose the alphabet ordering to be $\{b<c<a<d\}$, the Lyndon factorization of $\s{x}$ becomes $a \ge bcabcdabcaba$. 

Towards this goal we now describe a greedy algorithm for producing small numbers of factors which has performed well in practice on the biological $\{A, C, G,\\ 
T\}$ alphabet -- the experimentation compared results with those for the 4! letter permutations. Suppose the alphabet $\Sigma$ is size $\sigma$, and for a given string $\s{v} = v_1 \ldots v_n$, further suppose that the number of distinct characters in $\s{v}$ is $\delta \le \sigma$; for practical purposes we can assume $\sigma = n$.  

The proposed method requires an extension to a Parikh vector, $p(\s{v})$,
of a finite word \s{v}, where $p(\s{v})$ enumerates the 
occurrences 
of each letter of the alphabet in \s{v}. 
Our modification is that for each distinct letter we will record its individual RLE (run length encoding) exponent pattern -- so the sum of these exponents is the Parikh entry for that letter. We call this the Exponent Parikh vector, or EP vector.
For example, over the 
alphabet $\Sigma = \{b<c<d<f\}$, if $\s{v} = bbbffbbcf$ then $p (\s{v}) = [5,1,0,3]$; whereas, for the  EP vector we record the strings [(32), (21), (1)]. So usually the letters are listed in alphabetical order with a Parikh vector while in the EP case we are listing them in order of first occurrence.  

An overview of the method is that we use the fact that in a Lyndon factorization 
the first factor is the longest prefix which is a Lyndon word. Then the heuristic is that the left-most letter, $\alpha$ say, in the given string whose exponents form a Lyndon word with the minimal number of factors is chosen as the least letter in the alphabet ordering. In order to respect the Lyndon property for letters via their exponents, we require the exponent integer alphabet to be inverted, that is let $\bar{\Sigma} = \{ \ldots 3 < 2 < 1 \}$. Next, the algorithm attempts to assign order to letters in the 
substrings between runs of $\alpha$ characters, where these substrings are denoted $X_{i}$  -- if it gets stuck it tries backtracking.  

So note that with this algorithm the required property for the exponents of $\alpha$ is that they form a Lyndon word over $\bar{\Sigma}$ and in conjunction a requirement for assigning letters to the $ X_{i}$ substrings is that the ordering will be cycle-free. The algorithm can be modified to generate large numbers of factors which involves assigning different letters to be in decreasing order.

\subsection{Greedy algorithm}
\label{greedy-alg}

The pseudocode in Algorithm \ref{alg-Lynfactors} greedily assigns an alphabet order to letters. 

\begin{algorithm}
\caption{Order the alphabet so as to reduce the number of factors in a Lyndon factorization.}
\label{alg-Lynfactors}

\DontPrintSemicolon

With a linear scan record the Exponent Parikh (EP) vector of the string for $\delta$ distinct letters -- $O(n)$ \;

Compute $F_{\mathcal{L}}(\s{p_r})$ of each exponent string \s{p_r} over $\bar{\Sigma}$ and record its number of factors -- $O(n)$ \;

\While { bool = true } {

Select the next leftmost $\s{p_r}$, $\s{p_i}$ say, with minimal number of factors, $t$ say -- $O(n)$ \; 

\tcp{assign alphabet order to the $t$ factors of $F_{\mathcal{L}}(\s{p_i}) = \s{f_1} \ge \cdots \ge \s{f_t}$}  

\tcp{where $\s{f_j} = \alpha^{j_1} X_{1} \alpha^{j_2} X_{2} \cdots \alpha^{j_q} X_{q}$, and $\alpha \notin X_{h}$, $1 \le h \le q$, with $j_{1} j_{2} \ldots j_{q} \in \mathcal{L}$ over $\bar{\Sigma}$}



$\alpha = \lambda_1$  \tcp*[f]{assign first letter to be minimal in $\Sigma$; if $q=1$ assign each new letter in $X_1$ successively in $\Sigma$}

\For { $h = 2$ to $q$ }  
{
    \If (\tcp*[f]{same exponents so assign alphabet in order to letters in $X_1$ and $X_h$ substrings})  { $j_h = j_1$ } 
    { 
         $d \gets 1$ \;
         \While {$X_{1}[d] = X_{h}[d]$ } {
            assign each new letter successively in $\Sigma$; d++; \; 
         }

         \uIf { $X_{1}[d] = \alpha \And X_{h}[d] \neq \alpha$ } 
              { assign $X_{h}[d]$ to be next successive letter \; }
         
         \uElseIf { $X_{h}[d] = \alpha \And X_{1}[d] \neq \alpha$ } 
              { bool = false    \tcp*[r] {not Lyndon} }
         \uElseIf(\tcp*[f]{$X_{1}[d] \neq X_{h}[d]$ } ){assignment would not make inconsistency}  
              {assign $X_{h}[d] > X_{1}[d]$ \; }
             
         \Else 
            {bool = false  \tcp*[r] {inconsistent} }

    }
  
}

}

\uIf {bool} {
  {attempt assignment process for the $t$ factors of $F_{\mathcal{L}}(\s{p_i})$}  
}

\uIf {bool} {
  \If(\tcp*[f]{$\alpha \notin \s{u}$ }) { string prefix \s{u} (prior to \s{f_1}) is non-empty } 
     { repeat process on \s{u} starting with next successive letters in $\Sigma$ \; \tcp*[h]{lookup EP vector} 

{complete assignment of any remaining letters;}

{if letters in prefix \s{u} do not occur in suffix then re-assign all letters starting from prefix}}

  }
\Else {  
   {arbitrarily choose next leftmost ${p_r}$ with minimal number of factors and attempt new assignment \;}
   }

\end{algorithm}

\noindent The following example illustrates how backtracking can
lead the algorithm from an inconsistent ordering to a successful assignment and associated factorization.

\begin{example}
\label{ex-bcktrck} 
Assume $\Sigma = \{a,b,c,d\}$ and $\s{x} 
= a^2 bdc a^2 cd a^2 bdb a^1 b  a^2 bc a^2 c a^2 c a^1 b$. Only the letter $a$ has an exponent greater than 1 and $F_{\mathcal{L}}(EP(a)) = 2221 \ge 2221$ with $F_{\mathcal{L}}(\s{p_1}) = F_{\mathcal{L}}(\s{p_2}) = 2221$. 
Choosing 
$\s{p_1}$ causes inconsistency and similarly $\s{p_2}$.
So the algorithm then backtracks through the EP array and chooses the letter with the least number of factors (albeit singletons) -- the result is $\Sigma = \{d < c < a < b\}$ with $F_{\mathcal{L}}(\s{x}) = aab \ge dc aacd aabdb ab  aabc aac aac ab$.
\end{example}

\section{Applications}
\label{apps}

\par In many cases, such as natural language text processing, the order of the alphabet is prescribed, and hence the Lyndon factors of an input text cannot be manipulated. 
On the other hand, bioinformatics alphabets have no inherent ordering suggested by biological systems and applications involving Lyndon words, such as the Burrows-Wheeler transform (BWT), will allow for useful manipulation of the Lyndon factors. The co-BWT is the regular BWT of the reversed string, or the BWT with co-lexorder, which has been applied in the highly successful Bowtie sequence alignment program \cite{LTPS-09}. Integral with the BWT transform is the computation of suffix arrays via induced suffix-sorting. We also propose that pattern matching can be implemented with the Lyndon factorization in big data applications, such as sequence alignment, and further enhanced by fortuitous arrangements of the alphabet.

\par We refer to \cite{MRRS14} where a new method is presented for constructing the suffix array of a text by  using its Lyndon factorization advantageously. Partitioning the text according to its Lyndon properties allows tackling the problem in local portions of the text, local suffixes, prior to extending the solution globally, to achieve the suffix array of the entire text -- the local portions are determined by the Lyndon factors. The algorithm iteratively finds a Lyndon factor, constructs its suffix array and merges the new local suffixes with the current alphabetical list of suffixes. It is stated that the algorithm's time complexity is not competitive for the construction of the overall suffix array -- we propose that reducing the number of factors by alphabet ordering will improve the efficiency in practice. This is worthwhile as their algorithm offers flexibility by easily adapting to different implementations: 
online, external \& internal memory, and parallel.

\section{Experimentation: Factorization of DNA strings}

We chose as an example the 120 prokaryotic reference genomes from RefSeq\footnote{\url{https://ww.ncbi.nlm.nih.gov/refseq/about/prokaryotes}}, to investigate the results of the algorithm in practice\footnote{Code available at \url{https://github.com/amandaclare/lyndon-factors}}. Most of these genomes are provided as a single contiguous sequence but some of them have additional smaller pieces representing plasmids or other information. The longest contiguous sequence was chosen for each genome in these cases, and smaller pieces were discarded. The retained sequences ranged from 640,681 letters to 10,236,715 in length, with a mean of 3,629,792.

In order to determine how often our greedy algorithm found a good or optimal alphabet reordering in practice, we calculated the Lyndon factorizations resulting from all possible ($4! = 24$) alphabet reorderings of the characters A, C, G and T across this collection of genomes. The improvement that could potentially be made to the factorization by reordering is substantial, with at least a halving of the number of factors in most cases and an improvement reducing 25 factors down to 3 
in one case. For each genome we ranked the results of all possible reorderings by the number of factors produced and 
compared the reordering produced by the algorithm. The algorithm found the optimal reordering for 
21/120 genomes and the second-most optimal in 
31/120 genomes.

The EP vector is used to determine the least letter in the reordering. If the first choice leads to inconsistency (and hence small factors), backtracking to inspect other possible choices can be helpful. However, in many cases, the initial choice is still better than the next possible consistent solution found via backtracking. Without backtracking, the algorithm found 23/120 optimal orderings and a further second-most optimal orderings in 31/120 genomes.


\section{Research Problems}
\label{subsec-probs}

We propose the following research directions:

\begin{itemize}
\label{item-problems}
\item As a complementary structure to the Lyndon array 
we introduce and propose studies of the Lyndon factorization array. 
The \itbf{Lyndon array} $\s{\lambda} = \s{\lambda}_{\s{x}}[1..n]$
of a given $\s{x} = \s{x}[1..n]$
gives at each position $i$ the length
of the longest 
Lyndon word starting at $i$.
So we define the \itbf {Lyndon factorization array} $ \s{F} = \s{F}_{\s{x}}[1..n]$ of $\s{x}$ 
to give at each position $i$ the number of factors in the 
Lyndon factorization starting at $i$.

\item The greedy algorithm presented here does not necessarily produce an optimal solution hence natural problems are to design algorithms 
for Lyndon factorizations with a guaranteed minimal / maximal number of Lyndon factors. The optimization problem can be stated for any other circ-UMFF.

\item Using Duval's Lyndon factorization algorithm \cite{Du-83} as a benchmark, modify the alphabet order so as to increase/decrease the number of factors.



\item Theorem \ref{thm-faxmono} supports the following problem from \cite{DDS-09}:
Given a string $\s{u}$, determine the circ-UMFF(s) which factorizes $\s{u}$ into the maximal or minimal number of factors -- this can be combined with alphabet ordering.

\end{itemize}

\noindent{\bf Acknowledgements} 


     \includegraphics[width=0.19 \linewidth] {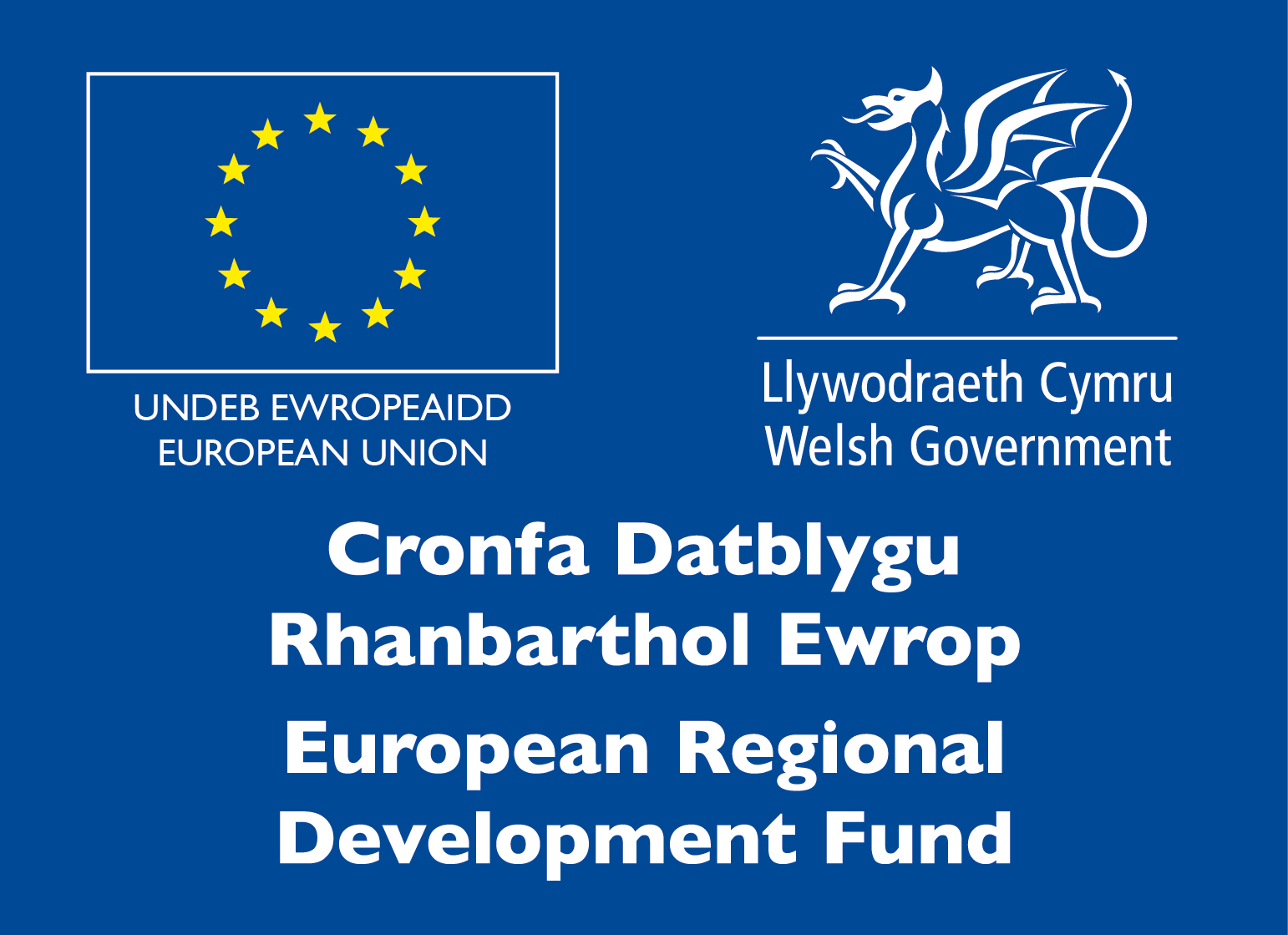}

\section*{References}

\bibliography{mybibfile}

\end{document}